\documentclass[11pt]{article}
\usepackage[utf8]{inputenc}
\usepackage{amsmath,amssymb,amsthm}
\usepackage{enumerate}
\usepackage{hyperref}
	
\newcommand{\abc}[2]{%
\mbox{$1|\,#1\,|\sum #2$}%
}

\def\eg{{e.g.}}
\def\ie{{i.e.}}
\def\st{\text{s.t.}}

\def\setS{S}

\newtheorem{theorem}{Theorem}
\newtheorem{lemma}{Lemma}
\newtheorem{corollary}{Corollary}
\newtheorem{conjecture}{Conjecture}

\title{A 2.542-Approximation for Precedence Constrained Single Machine Scheduling with Release Dates and Total Weighted Completion Time Objective}
\author{Martin Skutella\footnote{TU Berlin, Institut f\"ur Mathematik, MA 5-2, Stra{\ss}e des 17.~Juni 136, 10623 Berlin, Germany, \href{mailto:martin.skutella@tu-berlin.de}{\texttt{martin.skutella@tu-berlin.de}}}
}

\begin{document}
\maketitle

\begin{abstract}
We present a $\sqrt{e}/(\sqrt{e}-1)$-approximation algorithm for the nonpreemptive scheduling problem to minimize the total weighted completion time of jobs on a single machine subject to release dates and precedence constraints. The previously best known approximation algorithm dates back to 1997; its performance guarantee can be made arbitrarily close to the Euler constant~$e$~\cite{SS-RANDOM97}.
\end{abstract}

\section{Introduction}

We consider the following classical machine scheduling problem denoted by \abc{r_j,\,prec}{w_jC_j} in the standard classification scheme of Graham, Lawler, Lenstra, and Rinnooy Kan~\cite{GLLRK79}. We are given a set of jobs~$N=\{1,2,\dots,n\}$ and for every job $j\in N$ a processing time~$p_j\geq0$, a release date~$r_j\geq0$, and a weight~$w_j\geq0$. The jobs $j\in N$ need to be processed during non-overlapping time intervals of length $p_j$, and $j$'s processing must not start before its release date~$r_j$. Moreover, there are precedence constraints given by a partial order ``$\prec$'' on~$N$ where $j\prec k$ means that job~$j$ must be completed before job~$k$ may be started, that is, $j$'s processing interval must precede~$k$'s. We may therefore without loss of generality assume throughout the paper that $j\prec k$ implies $r_j\leq r_k$. The objective is to minimize the total weighted completion time~$\sum_{j\in N}w_jC_j$ where $C_j$ denotes the first point in time at which $j$'s processing is completed. 

\paragraph{Complexity.}
Even for unit job weights, the special cases of the problem without non-trivial release dates \abc{prec}{C_j} (\ie, $r_j=0$ for all $j\in N$) or without precedence constraints \abc{r_j}{C_j} are strongly NP-hard; see, \eg,~\cite[problem~SS4]{GJ79}. In preemptive scheduling, the processing of a job may be repeatedly interrupted and resumed at a later point in time. In the absence of precedence constraints, the problem with unit job weights \abc{r_j,\,pmtn}{C_j} can be solved in polynomial time~\cite{Bak74}, but for arbitrary weights \abc{r_j,\,pmtn}{w_jC_j} is strongly NP-hard. Without non-trivial release dates preemptions are superfluous such that \abc{prec,\,pmtn}{C_j} is equivalent to \abc{prec}{C_j} and thus strongly NP-hard.

\paragraph{List scheduling.}
Before dipping into the rich history of approximation algorithms for these scheduling problems, we first discuss the most important algorithmic ingredient for both heuristic and exact solutions: \emph{list scheduling}. Consider a list representing a total order on the set of jobs~$N$, extending the given partial order~``$\prec$''. A straightforward way to construct a feasible schedule is to process the jobs in the given order as early as possible with respect to release dates. A schedule constructed in this way is a \emph{list schedule}.

Depending on the given list and the release dates of jobs, the machine might remain idle when one job is completed but the next job in the list is not yet released.  On the other hand, if job preemptions are allowed, it is certainly not advisable to leave the machine idle while another job at a later position in the list is already available (released) and waiting. Instead, we better start this job and preempt it from the machine as soon as the next job in the list is released.  In \emph{preemptive list scheduling} we process at any point in time the first available job in the list.  The resulting preemptive schedule is feasible (as $j\prec k$ implies $r_j\leq r_k$) and is called \emph{preemptive list schedule}.  

\paragraph{Known techniques and results.}
There is a vast literature on approximation algorithms for the various scheduling problems mentioned above. Here we only mention those results that are particularly relevant in the context of this paper and refer to Chekuri and Khanna~\cite{ChekuriKhanna04} for a more comprehensive overview. Various kinds of linear programming (LP) relaxations have proved to be useful in designing approximation algorithms. One of the simplest and most intuitive classes of LP relaxations is based on completion time variables only. These LP relaxations were introduced by Queyranne~\cite{Queyranne93} and first used in the context of approximation algorithms by Schulz~\cite{schulz96}, who presents a $2$-approximation algorithm for the problem \abc{prec}{w_jC_j} and a $3$-approximation algorithm for \abc{r_j,\,prec}{w_jC_j}; see also Hall, Schulz, Shmoys, and Wein~\cite{HSSW97}. These algorithms compute an optimal LP solution and then do list scheduling in order of increasing LP completion times. Moreover, Hall et al.~\cite{HSSW97} show that preemptive list scheduling in order of increasing LP completion times is a $2$-approximation algorithm for \abc{r_j,\,prec,\,pmtn}{w_jC_j}.

Phillips, Stein, and Wein~\cite{PSW98} and Hall, Shmoys, and Wein~\cite{HSW96} introduce the idea of list scheduling in order of so-called $\alpha$-points to convert preemptive schedules to nonpreemptive ones. For $\alpha\in(0,1]$, the $\alpha$-point of a job with respect to a preemptive schedule is the first point in time when an $\alpha$-fraction of the job has been completed. Goemans~\cite{Goemans97} and Chekuri, Motwani, Natarajan, and Stein~\cite{CMNS-SICOMP} show that choosing $\alpha$ randomly leads to better results. In particular, Chekuri et al.~\cite{CMNS-SICOMP} present an $e/(e-1)$-approximation algorithm for \abc{r_j}{C_j} by starting from an optimal preemptive schedule. Goemans~\cite{Goemans97} and Goemans, Queyranne, Schulz, Skutella, and Wang~\cite{GQSSW02} give approximation results for the more general weighted problem \abc{r_j}{w_jC_j} based on a preemptive schedule that is an optimal solution to an LP relaxation in time-indexed variables. Similarly, Schulz and Skutella~\cite{SS-RANDOM97} give an $(e+\varepsilon)$-approximation algorithm for \abc{r_j,\,prec}{w_jC_j} for any~$\varepsilon>0$.

Bansal and Khot prove in a recent landmark paper~\cite{BansalKhot-FOCS2009} that there is no $(2-\varepsilon)$-approximation algorithm for \abc{prec}{w_jC_j}, assuming a stronger version of the Unique Games Conjecture. Amb\"uhl, Mastrolilli, Mutsanas, and Svensson~\cite{AmbuehlMMS2011}, based on earlier work of Correa and Schulz~\cite{CorreaSchulz2005} and Amb\"uhl and Mastrolilli~\cite{AmbuehlMastrolilli2009}, prove an interesting relation between the approximability of \abc{prec}{w_jC_j} and the vertex cover problem

\paragraph{Our contribution.}
We present a $\sqrt{e}/(\sqrt{e}-1)$-approximation algorithm for \abc{r_j,\,prec}{w_jC_j} based on the following two ingredients: (i) For the problem \abc{r_j,\,prec,\,pmtn}{w_jC_j} we slightly strengthen the $2$-approxima\-tion result of Hall et al.~\cite{HSSW97} and show that preemptive list scheduling in order of increasing LP completion times on a machine running at double speed yields a schedule whose cost is at most the cost of an optimal schedule on a regular machine; see Section~\ref{sec:resource-augmentation}. (ii) Modifying the analysis of Chekuri et al.~\cite{CMNS-SICOMP} we show how to turn the preemptive schedule on the double speed machine into a nonpreemptive schedule on a regular machine while increasing the objective function by at most a factor of $\sqrt{e}/(\sqrt{e}-1)$; see Section~\ref{sec:alpha-points}. We conclude with a conjecture in Section~\ref{sec:conclusion}.

\section{Optimal preemptive schedules under resource augmentation}
\label{sec:resource-augmentation}

In this section we consider the preemptive single machine scheduling problem with release dates, precedence constraints and total weighted completion time objective \abc{r_j,\,prec,\,pmtn}{w_jC_j}. The best known approximation result for this problem is a $2$-approximation algorithm due to Hall et al.~\cite{HSSW97} that is based on an LP relaxation in completion time variables originally introduced by Queyranne~\cite{Queyranne93} and later refined by Goemans~\cite{Goemans96,Goemans97} for problems involving release dates.
Let $S\subseteq N$ denote a set of jobs and define
\[
  p(S):=\sum_{j\in S} p_j \qquad\mbox{ and }\qquad
  r_{\min}(S):=\min_{j\in S} r_j.
\]
The LP relaxation in completion time variables $C_j$, $j\in N$, looks as follows:
\begin{align}
\min\quad&\sum_{j\in N} w_jC_j\notag\\
\st\quad&C_j\leq C_k && \text{for all $j\prec k$,}\label{eq:precedences}\\
&\frac1{p(S)}\sum_{j\in S} p_jC_j\geq r_{\min}(S)+\tfrac12p(S) && \text{for all~$S\subseteq N$.}\label{eq:parallel} 
\end{align}
Notice that constraints~\eqref{eq:precedences} could be strengthened to $C_j+p_k\leq C_k$, which is however not necessary for our purposes. Goemans~\cite{Goemans97} argues that constraints~\eqref{eq:parallel} hold for a feasible schedule, even if $(C_j)_{j\in N}$ denotes the vector of \emph{mean busy times} of jobs instead of the larger completion times. Moreover, despite their exponential number, these constraints can be separated in polynomial time by efficient submodular function minimization~\cite{Goemans96}. Thus, an optimal solution~$C^*$ to the LP relaxation can be found in polynomial time and yields the LP lower bound $\sum_{j\in N}w_jC^*_j$ on the total weighted completion time of an optimal preemptive schedule. Reindex the set of jobs such that
\begin{align}
C^*_1\leq C^*_2\leq\dots\leq C^*_n\qquad\text{and}\qquad (j\prec k	\Rightarrow j<k).\label{eq:list}
\end{align}
The second condition in~\eqref{eq:list} is necessary to ensure that the total order of jobs by increasing indices extends the partial order given by the precedence constraints; notice, that in an optimal LP solution~$C_j^*$ might be equal to~$C_k^*$ for some pair of jobs with $j\prec k$.

Hall et al.~\cite{HSSW97} show that preemptive list scheduling according to list~\eqref{eq:list} yields a feasible preemptive schedule with completion times~$C_j\leq 2\cdot C^*_j$, $j\in N$, and thus a $2$-approximate solution. Exactly the same analysis implies a slightly stronger result in terms of resource augmentation as we show in the next lemma. We imagine a machine running at double speed such that each job $j\in N$ needs to be processed for $p_j/2$ time units only. 

\begin{lemma}\label{lem:pmtn-list-sched}
Preemptive list scheduling according to list~\eqref{eq:list} on a machine running at double speed yields a feasible preemptive schedule with completion times
$C'_j\leq C^*_j$ for all $j\in N$.	
\end{lemma}

\begin{proof}
For a fixed $j\in N$, let $\setS$ denote the subset of jobs~$k\leq j$ such that (i)~$C'_k\leq C'_j$, (ii) the preemptive list schedule does not leave the double speed machine idle between times $C'_k$ and $C'_j$, and (iii) only jobs $\ell\leq j$ are being processed between times $C'_k$ and $C'_j$. 

\emph{Claim:} The preemptive list schedule processes the set of jobs~$\setS$ without interruption during the time interval~$I:=[r_{\min}(\setS),C_j']$. 

To prove the claim, consider a job $h\in\setS$ with~$r_h=r_{\min}(\setS)$. Since $\ell\prec h$ implies $r_{\ell}\leq r_h$, there is no idle time within the time interval $[r_h,C_h']$ and only jobs $\ell\leq h\leq j$ are being processed there. Moreover, due to (ii) there is no idle time within the time interval $[C_h',C_j']$ and only jobs $\ell\leq j$ are being processed there due to~(iii). As a consequence, there is no idle time in $I$ and only jobs $\ell\leq j$ are being processed there. Therefore, every job~$k$ with $C_k'\in I$ satisfies conditions (i), (ii), and (iii), and is thus contained in~$\setS$. Finally, since any job $\ell$ that the preemptive list schedule processes in $I$ satisfies $\ell\leq j$ and is therefore completed before job~$j$ in~$I$, the claim follows.

The claim implies that $C_j'=r_{\min}(\setS)+\frac12p(\setS)$. Finally,
\begin{align*}
C^*_j\geq\frac1{p(\setS)}\sum_{k\in\setS}p_kC^*_k\geq r_{\min}(\setS)+\tfrac12 p(\setS)=C'_j,
\end{align*}
where the first inequality holds as $C_k^*\leq C_j^*$ for $k\leq j$ and the second inequality follows by the LP constraints~\eqref{eq:parallel}.
\end{proof}

With the help of Lemma~\ref{lem:pmtn-list-sched} we can now prove the main result of this section.

\begin{theorem}\label{thm:resource-augmentation}
For a single machine running at double speed one can obtain in polynomial time a preemptive list schedule whose total weighted completion time is at most the LP lower bound $\sum_{j\in N}w_jC^*_j$ on the optimal total weighted completion time of a preemptive schedule for a regular single machine.
\end{theorem}

\begin{proof}
As discussed above, an optimal solution~$C^*$ to the LP relaxation can be obtained in polynomial time. Our algorithm then applies preemptive list scheduling according to list~\eqref{eq:list} on a double speed machine. By Lemma~\ref{lem:pmtn-list-sched}, the total weighted completion time of the resulting preemptive list schedule is bounded from above by the LP lower bound $\sum_{j\in N}w_jC^*_j$.
\end{proof}

\section{Scheduling in order of alpha-points}
\label{sec:alpha-points}

In this section we show how to turn a preemptive schedule on the double speed machine into a nonpreemptive schedule on a regular machine while increasing the total weighted completion time by a factor at most~$2.542$. 

\begin{theorem}\label{thm:alpha-points}
Given a feasible preemptive list schedule $S'$ on a double speed machine with completion times $C'_j$, $j\in N$, one can obtain in polynomial time a feasible nonpreemptive schedule on a regular speed machine with total weighted completion time 
\begin{align*}
\sum_{j\in N}w_jC_j\leq\frac{\sqrt{e}}{\sqrt{e}-1}\sum_{j\in N}w_jC'_j.
\end{align*}
\end{theorem}

Theorems~\ref{thm:resource-augmentation} and~\ref{thm:alpha-points} together yield the new approximation result for the scheduling problem under consideration.

The proof of Theorem~\ref{thm:alpha-points} relies on list scheduling in order of $\alpha$-points: For $0<\alpha\leq1$, the $\alpha$-point $C'_j(\alpha)$ of job $j$ with respect to schedule $S'$ is the first point in time when job $j$ has been processed for $\alpha\cdot p_j/2$ time on the double speed machine. Consider the list schedule $S_{\alpha}$ obtained by scheduling jobs in order of increasing $C'_j(\alpha)$ on a regular speed machine (notice that this order is in line with the precedence constraints as the preemptive schedule~$S'$ is feasible). Let $C_j^{\alpha}$ denote job $j$'s completion time in the list schedule $S_{\alpha}$. Moreover, for a fixed job~$k$, let $\eta_j$ denote the fraction of job $j\in N$ that has been processed in schedule $S'$ on the double speed machine by time $C'_k$. In particular,
\begin{align}
C'_k\geq\sum_{j\in N}\eta_j\frac{p_j}2.\label{eq:lower-bound-C_k}
\end{align}

The following lemma is a slight modifications of a more general observations presented in~\cite[Corollary~2.3.3]{Skut-Diss98} and~\cite[Corollary~3.3]{Skut-APPOL}. We give a new, somewhat simpler proof.

\begin{lemma}\label{lem:alpha-points}
\begin{align}
C_k^{\alpha}\leq C'_k+\sum_{j:\eta_j\geq\alpha}\Bigl(1+\frac{\alpha-\eta_j}2\Bigr)p_j.\label{eq:alpha-points}
\end{align}	
\end{lemma}

\begin{proof}
Let $X$ denote the subset of all jobs~$j$ scheduled by the list schedule~$S_{\alpha}$ no later than job~$k$ such that there is no idle time between times~$C_j^{\alpha}$ and~$C_k^{\alpha}$. In particular, $k\in X$. Moreover, let $\ell\in X$ denote the job that $S_{\alpha}$ schedules first among the jobs in~$X$. By definition of~$X$, job~$\ell$ is started at its release date~$r_{\ell}$. In particular,
\begin{align}
C_k^{\alpha}=r_{\ell}+\sum_{j\in X}p_j.\label{eq:alpha1}
\end{align} 
Since schedule~$S'$ obeys release dates, we get~$r_{\ell}\leq C_{\ell}'(\alpha)\leq C_j'(\alpha)$ for every job~$j\in X$. Moreover, by definition of~$\eta_j$, schedule~$S'$ processes every job~$j\in X$ in the time interval $[C_j'(\alpha),C_k']$ for exactly $\frac12(\eta_j-\alpha)p_j$ time units on the double speed machine. Since the machine can process at most one job at a time,
\begin{align}
C_k'\geq r_{\ell}+\sum_{j\in X}\frac{\eta_j-\alpha}2 p_j.\label{eq:alpha2}
\end{align}
Combining~\eqref{eq:alpha1} and~\eqref{eq:alpha2} yields
\begin{align}
C_k^{\alpha}\leq C_k'+\sum_{j\in X}\Bigl(1+\frac{\alpha-\eta_j}2\Bigr)p_j.\label{eq:alpha3}
\end{align}
Notice that the sum on the right hand side of~\eqref{eq:alpha-points} goes over a superset of~$X$ since $C_j'(\alpha)\leq C_k'(\alpha)\leq C_k'$ and thus $\eta_j\geq\alpha$ for every job~$j\in X$. Finally, as $\alpha>0$ and $\eta_j\leq1$, the summand of every job~$j$ in~\eqref{eq:alpha-points} is at least~$p_j/2$ and thus nonnegative. Therefore inequality~\eqref{eq:alpha-points} follows immediately from~\eqref{eq:alpha3}.
\end{proof}

We now draw $\alpha$ randomly from $(0,1]$ with density function
\begin{align*}
f(\alpha):=\frac{e^{\alpha/2}}{2(\sqrt{e}-1)}.
\end{align*}
Notice that for $0\leq\eta\leq1$
\begin{align*}
\int_0^{\eta}f(\alpha)\Bigl(1+\frac{\alpha-\eta}2\Bigr)\,\text{d}\alpha	= \frac{\eta}{2(\sqrt{e}-1)}.
\end{align*}
Thus, by Lemma~\ref{lem:alpha-points},
\begin{align*}
\text{E}[C_k^{\alpha}]&\leq C'_k+\sum_{j\in N}p_j\int_0^{\eta_j}f(\alpha)\Bigl(1+\frac{\alpha-\eta_j}2\Bigr)\,\text{d}\alpha\\
&= C'_k+\frac1{\sqrt{e}-1}\sum_{j\in N}\eta_j \frac{p_j}2\leq\frac{\sqrt{e}}{\sqrt{e}-1}C'_k,
\end{align*}
where the last inequality follows from~\eqref{eq:lower-bound-C_k}. By linearity of expectation, the expected total weighted completion time of the nonpreemptive list schedule in order of random $\alpha$-points is at most $\sqrt{e}/(\sqrt{e}-1)$ times larger than the total weighted completion time of the given preemptive schedule~$S'$.

To complete the proof of Theorem~\ref{thm:alpha-points} we need to derandomize the choice of~$\alpha$. In a preemptive list schedule, preemption of a job can only occur when another job is released. In particular, there can be at most $n-1$ preemptions and therefore at most $n$ combinatorially different choices of~$\alpha$. As observed by Goemans~\cite{Goemans97}, an~$\alpha$ minimizing the total weighted completion time of the resulting list schedule can thus be found in $O(n^2)$ time. 


As a consequence of Theorems~\ref{thm:resource-augmentation} and~\ref{thm:alpha-points} we can state the following Corollary on the quality of the lower bound given by an optimal solution to the LP relaxation in Section~\ref{sec:resource-augmentation}.

\begin{corollary}\label{cor:LP-quality}
For instances of the scheduling problem \abc{r_j,\,prec}{w_jC_j}, the value of an optimal solution to the LP relaxation in completion time variables considered in Section~\ref{sec:resource-augmentation} is at most a factor~$(\sqrt{e}-1)/\sqrt{e}$ smaller than the total weighted completion time of an optimal schedule.
\end{corollary}

\begin{proof}
By Theorems~\ref{thm:alpha-points} and~\ref{thm:resource-augmentation} our approximation algorithm constructs a feasible schedule whose total weighted completion time is bounded by
\begin{align}
\sum_{j\in N}w_jC_j\leq\frac{\sqrt{e}}{\sqrt{e}-1}\sum_{j\in N}w_jC'_j\leq\frac{\sqrt{e}}{\sqrt{e}-1}\sum_{j\in N}w_jC^*_j,\label{eq:LP-quality}
\end{align}
where $C'$ denotes the vector of completion times of the preemptive double speed machine schedule $S'$ and $C^*$ is an optimal LP solution. Since the left-hand side of~\eqref{eq:LP-quality} is an upper bound on the total weighted completion time of an optimal schedule, the result follows.
\end{proof}

\section{Concluding remarks}
\label{sec:conclusion}

Despite our enthusiastic yet ultimately fruitless efforts to improve the presented approximation result, we feel that the new performance guarantee \mbox{$\sqrt{e}/(\sqrt{e}-1)$} is hardly the last word on the considered scheduling problem. On the other hand, the history of approximation algorithms for the special case \abc{prec}{w_jC_j} and, in particular, recent non-approximability results make it seem somewhat unlikely to achieve a performance ratio strictly better than~$2$. Therefore, and due lack of imagination of other meaningful approximation ratios, we conclude with the following conjecture, granting an extra $+\varepsilon$ in the performance ratio to the release dates.

\begin{conjecture}
For any $\varepsilon>0$, there is a $(2+\varepsilon)$-approximation algorithm for \abc{r_j,\,prec}{w_jC_j}.
\end{conjecture}

\paragraph{Acknowledgements.} We are grateful to an anonymous referee for various helpful comments that have lead to an improved presentation of the paper. This work is supported by the Einstein Foundation Berlin.

\footnotesize
\bibliographystyle{plain}
\bibliography{../BibTeX/mybib.bib}

\begin{thebibliography}{10}

\bibitem{AmbuehlMastrolilli2009}
C.~Amb\"uhl and M.~Mastrolilli.
\newblock Single machine precedence constrained scheduling is a vertex cover
  problem.
\newblock {\em Algorithmica}, 53:488--503, 2009.

\bibitem{AmbuehlMMS2011}
C.~Amb\"uhl, M.~Mastrolilli, N.~Mutsanas, and O.~Svensson.
\newblock On the approximability of single-machine scheduling with precedence
  constraints.
\newblock {\em Mathematics of Operations Research}, 36:653--669, 2011.

\bibitem{Bak74}
K.~R. Baker.
\newblock {\em Introduction to Sequencing and Scheduling}.
\newblock John Wiley \& Sons, New York, 1974.

\bibitem{BansalKhot-FOCS2009}
N.~Bansal and S.~Khot.
\newblock Optimal long code test with one free bit.
\newblock In {\em Proceedings of the 50th Annual IEEE Symposium on Foundations
  of Computer Science}, pages 453--462, 2009.

\bibitem{ChekuriKhanna04}
C.~Chekuri and S.~Khanna.
\newblock Approximation algorithms for minimizing average weighted completion
  time.
\newblock In J.~Leung, editor, {\em Handbook of Scheduling: Algorithms, Models,
  and Performance Analysis}. CRC Press, 2004.

\bibitem{CMNS-SICOMP}
C.~S. Chekuri, R.~Motwani, B.~Natarajan, and C.~Stein.
\newblock Approximation techniques for average completion time scheduling.
\newblock {\em SIAM Journal on Computing}, 31:146--166, 2001.

\bibitem{CorreaSchulz2005}
J.~R. Correa and A.~S. Schulz.
\newblock Single machine scheduling with precedence constraint.
\newblock {\em Mathematics of Operations Research}, 30:1005--1021, 2005.

\bibitem{GJ79}
M.~R. Garey and D.~S. Johnson.
\newblock {\em Computers and Intractability: A Guide to the Theory of
  {NP}{--}Completeness}.
\newblock Freeman, San Francisco, 1979.

\bibitem{Goemans96}
M.~X. Goemans.
\newblock A supermodular relaxation for scheduling with release dates.
\newblock In W.~H. Cunningham, S.~T. McCormick, and M.~Queyranne, editors, {\em
  Integer Programming and Combinatorial Optimization}, volume 1084 of {\em
  Lecture Notes in Computer Science}, pages 288--300. Springer, 1996.

\bibitem{Goemans97}
M.~X. Goemans.
\newblock Improved approximation algorithms for scheduling with release dates.
\newblock In {\em Proceedings of the 8th Annual {ACM--SIAM} Symposium on
  Discrete Algorithms}, pages 591--598, New Orleans, LA, 1997.

\bibitem{GQSSW02}
M.~X. Goemans, M.~Queyranne, A.~S. Schulz, M.~Skutella, and Y.~Wang.
\newblock Single machine scheduling with release dates.
\newblock {\em SIAM Journal on Discrete Mathematics}, 15:165--192, 2002.

\bibitem{GLLRK79}
R.~L. Graham, E.~L. Lawler, J.~K. Lenstra, and A.~H.~G. R{innooy Kan}.
\newblock Optimization and approximation in deterministic sequencing and
  scheduling: {A} survey.
\newblock {\em Annals of Discrete Mathematics}, 5:287--326, 1979.

\bibitem{HSSW97}
L.~A. Hall, A.~S. Schulz, D.~B. Shmoys, and J.~Wein.
\newblock Scheduling to minimize average completion time: Off-line and on-line
  approximation algorithms.
\newblock {\em Mathematics of Operations Research}, 22:513--544, 1997.

\bibitem{HSW96}
L.~A. Hall, D.~B. Shmoys, and J.~Wein.
\newblock Scheduling to minimize average completion time: {O}ff-line and
  on-line algorithms.
\newblock In {\em Proceedings of the 7th Annual {ACM--SIAM} Symposium on
  Discrete Algorithms}, pages 142--151, Atlanta, GA, 1996.

\bibitem{PSW98}
C.~Phillips, C.~Stein, and J.~Wein.
\newblock Minimizing average completion time in the presence of release dates.
\newblock {\em Mathematical Programming}, 82:199--223, 1998.

\bibitem{Queyranne93}
M.~Queyranne.
\newblock Structure of a simple scheduling polyhedron.
\newblock {\em Mathematical Programming}, 58:263--285, 1993.

\bibitem{schulz96}
A.~S. Schulz.
\newblock Scheduling to minimize total weighted completion time: {P}erformance
  guarantees of {LP}-based heuristics and lower bounds.
\newblock In W.~H. Cunningham, S.~T. McCormick, and M.~Queyranne, editors, {\em
  Integer Programming and Combinatorial Optimization}, volume 1084 of {\em
  Lecture Notes in Computer Science}, pages 301--315. Springer, 1996.

\bibitem{SS-RANDOM97}
A.~S. Schulz and M.~Skutella.
\newblock Random-based scheduling: {N}ew approximations and {LP} lower bounds.
\newblock In J.~Rolim, editor, {\em Randomization and Approximation Techniques
  in Computer Science}, volume 1269 of {\em Lecture Notes in Computer Science},
  pages 119--133. Springer, 1997.

\bibitem{Skut-Diss98}
M.~Skutella.
\newblock {\em Approximation and Randomization in Scheduling}.
\newblock PhD thesis, Technische Universit{\"a}t Berlin, Germany, 1998.

\bibitem{Skut-APPOL}
M.~Skutella.
\newblock List scheduling in order of $\alpha$-points on a single machine.
\newblock In E.~Bampis, K.~Jansen, and C.~Kenyon, editors, {\em Efficient
  Approximation and Online Algorithms: Recent Progress on Classical
  Combinatorial Optimization Problems and New Applications}, volume 3484 of
  {\em Lecture Notes in Computer Science}, pages 250--291. Springer, 2006.

\end{thebibliography}

\end{document}